\documentclass[11pt]{article}

\usepackage{amsfonts,amsthm,amssymb,amsmath,mathrsfs}
\usepackage{graphicx,color,hyperref,mathtools}
\usepackage[usenames,dvipsnames]{xcolor}
\usepackage{pdfsync}
\usepackage[margin=1.0in]{geometry}
\usepackage[normalem]{ulem}
\usepackage[T1]{fontenc}
\usepackage[scaled=.92]{helvet}
\usepackage[round]{natbib}
\usepackage{hyperref}
\usepackage{authblk}

\usepackage{multibib}
\newcites{A}{References for Appendices}
\usepackage{tabularx, booktabs, array, multirow,enumitem}
\usepackage{mdframed}

\usepackage{pgfplots}
\pgfplotsset{compat=1.18}
\usepackage{tikz}
\usetikzlibrary{calc,shapes,arrows.meta,pgfplots.groupplots,decorations.pathreplacing} 
\usepackage{subcaption}
\usepackage{graphicx}

\newtheorem{theorem}{Theorem}

\newtheorem{lemma}{Lemma}

\usepackage{xcolor}
\definecolor{lightgreen}{RGB}{110,170,110}
\usepackage{chngcntr}
\usepackage{apptools}
\AtAppendix{\counterwithin{figure}{section}}

\usepackage[english]{babel}
\usepackage[utf8]{inputenc}
\usepackage[colorinlistoftodos]{todonotes}
\usepackage{algorithm}
\usepackage{algpseudocode}

\usepackage{chngcntr}

\newcounter{parentnumber}

\usepackage[hang,flushmargin]{footmisc}

\usepackage[normalem]{ulem}
\allowdisplaybreaks

\usepackage{array}
\newcolumntype{P}[1]{>{\centering\arraybackslash}p{#1}}
\newcolumntype{M}[1]{>{\centering\arraybackslash}m{#1}}

\usepackage{amsmath}
\usepackage{amsfonts}
\usepackage{graphicx}
\usepackage{fancyvrb}
\usepackage{url}
\usepackage{listings}
\usepackage{setspace}
\usepackage{tikz}
\usetikzlibrary{automata,positioning}
\usepackage{pgfplots}

\title{Branching-Time Signatures of Growth Regime in Tumor Birth-Death Models}

\author{Meng Wang$^{1,\dagger}$ \quad \hspace*{-6pt} Zhangchi Weng$^{1,\dagger}$ \quad \hspace*{-6pt} Jasmine Foo$^{2}$ \quad \hspace*{-6pt}Zicheng Wang$^{1}$}
\date{%
    \footnotesize $^1$School of Data Science, The Chinese University of Hong Kong, Shenzhen, Guangdong 518172, China. \\[2pt]
     \footnotesize $^2$School of Mathematics, University of Minnesota, Minneapolis, MN 55455, USA. \\[2pt]
    $^{\dagger}$These authors contributed equally to this work 
}

\begin{document}

\maketitle

\begin{abstract}
Tumor evolution is shaped by cell division, cell death, competition, and constraints imposed by the local microenvironment. Because these dynamics are usually not observed directly, phylogenetic trees inferred from somatic variation in sampled tumor cells can provide an indirect record of the population history that produced the sample. In this paper, we examine whether the distribution of inferred internal branching times exhibits signatures that depend on the underlying tumor growth regime. Specifically, we study the distribution of internal branching times in continuous-time birth--death models of tumor evolution. Exponentially growing populations exhibit a unimodal distribution of internal branching times, with the mode located near the root. In contrast, logistic growth, which models expansion constrained by carrying capacity, yields a substantially more intricate genealogical structure: the distribution of branching times undergoes a systematic transition as the time elapsed since tumor initiation increases. Specifically, this progression shifts from an expansion-dominated phase, through an intermediate early–recent bimodal phase, to a final recent-dominated phase. Extensive simulations of reconstructed tumor genealogies support these theoretical findings.

\end{abstract}
\singlespacing


\thispagestyle{empty}

\newpage

\section{Introduction}

Biological populations often evolve under different growth regimes shaped by ecological and physical constraints. In microbial batch cultures, populations commonly pass from rapid expansion to slower growth or stationary phase as nutrients are depleted and waste products accumulate (\citealt{monod2012growth}). Viral populations can show analogous changes in growth rate. They may expand rapidly when susceptible hosts are abundant, and then grow more slowly or decline as susceptible hosts are depleted, immunity accumulates, or interventions are introduced (\citealt{grenfell2004unifying,stadler2013birth}). Tumor cell populations provide another important example. Early lesions or expanding clones often exhibit exponential growth, whereas later progression can be constrained by tissue architecture, limited space and nutrients, vascular supply, immune interactions, and other microenvironmental changes (\citealt{gerlee2013model,benzekry2014classical,noble2022spatial}).

Growth regime matters because two populations with the same size at the observation time may have different biological and genealogical histories. A population generated mainly by rapid expansion may have a different distribution of ancestry than a population that has persisted for a long time under density-dependent regulation. In the latter case, the sampled lineages may reflect a longer history of birth, death, competition, and turnover. These alternative growth dynamics can produce distinct patterns of genetic diversity, different probabilities that rare variants are maintained or lost, and different potential responses to environmental change or therapy. Classical population-genetic results show that stable and rapidly expanding populations can produce different distributions of pairwise sequence differences (\citealt{slatkin1991pairwise}). In cancer, clonal analyses and mathematical models suggest that tumor growth regime can shape clonal composition and intratumor heterogeneity (\citealt{driessens2012defining,sottoriva2015big,waclaw2015spatial,chkhaidze2019spatially,fu2022spatial,noble2022spatial,leder2025parameter}). Thus, understanding how growth regimes shape genetic and genealogical patterns is important for interpreting the biological state and possible future behavior of an observed population.

In many settings, however, the past population trajectory is not observed directly. Instead, one observes a sample of extant cells, individuals, or lineages at one or a few time points. In tumor evolution, advances in bulk, multi-region, and single-cell sequencing have made it possible to reconstruct evolutionary relationships among sampled tumor cells or clones (\citealt{navin2011tumour,gerstung2020evolutionary,lewinsohn2023state}). The resulting phylogenetic tree can retain partial information about the evolutionary history that generated the observed sample. Its topology describes ancestry relationships among sampled lineages. Its branch lengths, and its internal node times when the tree is time-calibrated, describe how inferred lineage divergences are distributed through time. These features provide a basis for studying how present-day samples reflect past population dynamics. In this paper, we focus on the timing of internal branching events in sampled genealogies. We investigate whether their distribution contains signatures associated with different growth regimes, especially exponential expansion and density-dependent, carrying-capacity-limited growth.

Throughout this paper, tumor evolution serves as the motivating example. However, the main results apply more broadly to evolving biological populations that can be modeled by birth--death processes. The remainder of the paper is organized as follows. Section~2 reviews the relevant literature. Section~3 introduces the birth–death models, sampled genealogies, and branching-time distribution. Section~4 presents our theoretical and simulation results. Section~5 offers concluding remarks. All proofs are provided in the Appendix.

\section{Related Literature}

Phylogeny-based approaches have long been used to study population history from reconstructed trees; see \citet{mooers1997inferring} and \citet{stadler2021phylodynamics} for reviews. A reconstructed phylogeny encodes more than shared ancestry among sampled individuals. When the tree is time-scaled, its internal node times represent inferred divergence times, or equivalently coalescence times when viewed backward from sampling. Its branch lengths describe the temporal separation among sampled lineages and their common ancestors. These temporal features underlie many tree-based methods that link sampled genealogies to the population process that generated the sample. Our work belongs to this broad literature and focuses specifically on a particular tree summary: the distribution, and in particular the modal structure, of internal branching times under prescribed birth–death growth regimes.

Coalescent theory provides a probabilistic link between the temporal structure of a genealogy and the population-size history that generated it. In a time-inhomogeneous Kingman coalescent, suppose that \(a(t)\) ancestral lineages are present at time \(t\) before sampling. The instantaneous rate at which some pair of lineages coalesces is proportional to \(\binom{a(t)}{2}/N_e(t)\), where \(N_e(t)\) denotes the effective population size. Thus, after accounting for the number of lineages at risk, intervals with many coalescent events correspond to higher coalescence rates and hence smaller effective population size. Conversely, long waiting times correspond to larger effective population size.\footnote{\(N_e(t)\) is a genealogical quantity that controls the rate of common ancestry, and it need not coincide with the census population size.} This connection has made coalescent theory a central tool for studying population-size change from genetic data and time-scaled genealogies
(\citealt{kingman1982coalescent,kuhner1998maximum,rosenberg2002genealogical,lambert2013birth,lambert2018coalescent,harris2020coalescent,cheek2022coalescent}).

Early coalescent-based methods used the timing of ancestral events to infer changes in population size. \citet{nee1995inferring} developed graphical approaches for detecting departures from constant population size, including exponential growth, from molecular trees. This idea was developed into skyline methods by \citet{pybus2000integrated} and \citet{strimmer2001exploring}. These methods estimate a piecewise-constant effective population-size trajectory from the waiting times between successive coalescent events, with the generalized skyline plot grouping adjacent intervals to reduce variance. Later Bayesian extensions, including the Bayesian skyline plot \citep{drummond2005bayesian} and the extended Bayesian skyline plot \citep{heled2008bayesian}, jointly estimate the genealogy, sequence-evolution parameters, and population-size history. They can also accommodate heterochronous sampling and multiple loci. Subsequent Bayesian nonparametric methods regularized these estimates by placing smoothing priors on the population-size trajectory. Examples include the Gaussian Markov random field priors used in the Bayesian skyride and skygrid \citep{minin2008smooth,gill2013improving}, Gaussian process priors \citep{palacios2013gaussian}, and horseshoe Markov random field priors designed to capture both smooth variation and abrupt changes \citep{faulkner2020horseshoe}.

Together, these methods show that the temporal distribution of coalescent events carries substantial information about time-varying effective population size. They are closely related to our work because both use internal node times as genealogical information. However, the emphasis is different. Skyline-type methods use these node times to reconstruct an effective population-size trajectory. This trajectory summarizes the combined effects of birth, death, sampling, and other processes through a single genealogical rate. By contrast, we start from explicit birth--death models and analyze how the shape of the internal branching-time distribution changes across growth regimes. In particular, we examine whether exponential expansion and density-dependent regulation leave different branching-time signatures in sampled genealogies.

A second related line of work uses birth--death and birth--death-sampling models directly. In epidemiological phylodynamics, birth--death skyline methods allow transmission, removal, and sampling parameters to vary through time, and use reconstructed phylogenies to estimate these time-varying quantities \citep{stadler2013birth}. Coalescent point process theory also provides tractable descriptions of reconstructed trees generated by birth--death processes, and is closely connected to the distribution of node depths in sampled trees \citep{lambert2013birth,lambert2018coalescent}. These approaches are relevant to our study because they connect forward-time branching processes with backward-time genealogical structure. Our contribution differs in emphasis. Rather than estimating a time-varying parameter trajectory from data, we characterize the branching-time distribution induced by specified exponential and density-dependent birth--death models.

Recent work has adapted phylogenetic and coalescent ideas to cellular and tumor evolution. \citet{johnson2023clonerate} developed coalescent-based estimators of clonal net growth rates from single-cell phylogenies or shared mutation patterns. In solid tumors, state-dependent phylodynamic models have been used to connect branching patterns, mutation accumulation, and spatial location in order to estimate spatial variation in cell-division rates \citep{lewinsohn2023state}. These studies demonstrate the growing role of tree-based methods in cancer and cell biology. Our work adds to this literature by asking a complementary question: how does the distribution of internal branching times change across growth regimes, and what qualitative signatures arise under exponential expansion versus density-dependent regulation?

\section{Model}\label{Sec: Model}

In this section, we specify the continuous-time birth--death models used to represent the growth regimes considered in this study: unconstrained exponential expansion and density-dependent logistic growth. We then describe the sampled genealogy induced by a finite sample of extant tumor cells under each model. Lastly, we define the internal branching times extracted from these sampled genealogies and the corresponding branching-time distribution, which is the main genealogical quantity analyzed in this paper.

\subsection{Birth--Death Model}

We model tumor-cell population dynamics as a continuous-time birth--death process. Let \(N(t)\) denote the number of living cells at time \(t\). Conditional on the current population size, each living cell experiences stochastic division and death events with rates determined by the specified growth regime. A division event replaces one cell by two daughter cells and therefore increases \(N(t)\) by one. A death event removes one living cell and decreases \(N(t)\) by one.

We write \(\lambda(n)\) and \(\mu(n)\) for the per-cell birth and death rates when the population size is \(n\). Thus, conditional on \(N(t)=n\), the process jumps from \(n\) to \(n+1\) at rate \(n\lambda(n)\), and from \(n\) to \(n-1\) at rate \(n\mu(n)\). In the following subsections, we specify these rates for the exponential-growth model and for two density-dependent logistic models.

\subsubsection{Exponential-Growth Model}

In the exponential-growth regime, the per-cell birth and death rates are constant and independent of population size:
\begin{equation}
    \lambda(n) \equiv \lambda_0, \qquad \mu(n) \equiv \mu_0 .
\end{equation}
We assume that \(\lambda_0 > \mu_0\), so that the birth--death process is supercritical. We define the net per-cell growth rate by
\begin{equation}
    r_0 = \lambda_0 - \mu_0 > 0 .
\end{equation}
For a deterministic initial population size \(N(0)=N_0\), the expected population size is $\mathbb{E}[N(t)] = N_0 e^{r_0t}$. Unless otherwise stated, we take \(N_0=1\), in which case \(\mathbb{E}[N(t)] = e^{r_0t}\). This constant-rate model represents unconstrained tumor-cell expansion. It provides a baseline approximation for early clonal growth, before density-dependent constraints such as limited space, resource limitation, crowding, and competition have appreciable effects.

\subsubsection{Density-Dependent Logistic Growth Models}

Following \citet{cheek2022coalescent}, we consider two density-dependent logistic birth--death models. Let \(\kappa\) denote the carrying capacity. Here, \(\kappa\) is the population scale at which the per-cell net growth rate is zero in the deterministic logistic approximation. It should not be interpreted as a hard upper bound on the stochastic population size.

The two logistic models have the same density-dependent net growth rate: $\lambda(n)-\mu(n)=r_0\left(1-\frac{n}{\kappa}\right)$. They differ, however, in how density dependence is assigned to the birth and death rates. In both models, \(\lambda_0\) and \(\mu_0\) denote the baseline low-density per-cell birth and death rates.

\textbf{Model 1: Density-dependent death.}
In the first model, density regulation acts through the death rate. The per-cell birth rate remains constant, whereas the death rate increases with population size:
\[
    \lambda(n)=\lambda_0,
    \qquad
    \mu(n)=\mu_0+\frac{n}{\kappa}(\lambda_0-\mu_0).
\]
Biologically, this formulation represents settings in which crowding, resource limitation, hypoxia, or other local environmental stresses increase cell death during later stages of tumor growth.

\textbf{Model 2: Density-dependent proliferation suppression.}
In the second model, density regulation acts through the birth rate. The per-cell death rate remains constant, whereas the birth rate decreases with population size:
\[
    \lambda(n)=\lambda_0-\frac{n}{\kappa}(\lambda_0-\mu_0),
    \qquad
    \mu(n)=\mu_0.
\]
Biologically, this formulation represents settings in which limited space, nutrient availability, growth signals, or other microenvironmental constraints suppress proliferation without directly increasing the baseline death rate.

Although the two logistic models have the same density-dependent net growth rate, they assign density regulation to different components of the birth--death process. As a result, they have different per-cell birth and death rates along the growth trajectory. These differences can affect the sampled genealogy even when the deterministic population-size trajectory is the same. We analyze both models to study how internal branching-time distributions reflect the transition from early expansion to density-dependent regulation, and to examine whether density-dependent death and density-dependent proliferation suppression produce different branching-time profiles.

\subsection{Sampled Phylogenetic Trees and Branching-Time Distribution}\label{Sec: Tree and Branching Time}

At a fixed observation time \(T\), let \(\mathcal{L}(T)\) denote the set of living tumor cells, so that $|\mathcal{L}(T)| = N(T)$. The complete population genealogy, including extinct lineages and unsampled side branches, is generally not observed. We therefore consider a finite sample of extant cells,
\[
    S = \{s_1,\ldots,s_m\} \subset \mathcal{L}(T),
\]
where \(m\leq N(T)\) is the sample size. Unless otherwise stated, the sampled cells are drawn uniformly without replacement from the extant population \(\mathcal{L}(T)\). 

The sampled phylogenetic tree associated with \(S\) is the rooted genealogy obtained by tracing the sampled cells backward in time to their most recent common ancestor and retaining only the ancestral lineages needed to connect them. The sampled cells form the leaves of the tree, and the root represents the most recent common ancestor of the sample. Internal nodes correspond to retained ancestral division events. These are division events for which both daughter lineages leave at least one descendant in \(S\). Thus, the sampled tree is a pruned representation of the full population genealogy. It provides a partial retrospective summary of the evolutionary history represented in the sample. Its topology records shared ancestry among sampled lineages, and its branch lengths record the elapsed time between retained ancestral events and sampled tips. The timing of internal nodes therefore describes how retained lineage-splitting events are distributed through the sampled genealogy.

For each sampled tree, we record the times of its internal branching nodes. Assume \(m\geq 2\), so that the sampled tree contains \(m-1\) internal branching nodes. Let
\[
    \mathcal{I}_S=\{1,\ldots,m-1\}
\]
denote the set of internal nodes in the sampled tree. For each \(i\in \mathcal{I}_S\), let \(\tau_{S,i}\) denote the forward-time location of internal node \(i\), measured on the same time scale as the birth--death process. Thus, $0 \leq \tau_{S,i} \leq T$. Under this convention, smaller values of \(\tau_{S,i}\) correspond to older branching events, closer to tumor initiation and to the root of the sampled tree. Larger values correspond to more recent branching events, closer to the sampled tips at time \(T\).

To compare branching times across observation times and across realizations, we scale internal branching times by the observation time \(T\). Specifically, we define
\[
    \widetilde{\tau}_{S,i}
    =
    \frac{\tau_{S,i}}{T},
    \qquad i\in \mathcal{I}_S .
\]
Then \(\widetilde{\tau}_{S,i}\in[0,1]\), where \(0\) corresponds to tumor initiation and \(1\) corresponds to the sampling time. Unlike normalization by the depth of each sampled tree, this scaling preserves the position of each branching event relative to the full observation interval \([0,T]\). In particular, the oldest branching event in the sampled tree need not have scaled time \(0\), because the most recent common ancestor of the sample may occur after tumor initiation.

The main genealogical quantity analyzed in this paper is the empirical distribution of the \(T\)-scaled internal branching times:
\[
    \widehat{F}_{S,T}(x)
    =
    \frac{1}{m-1}
    \sum_{i\in\mathcal{I}_S}
    \mathbf{1}\{\widetilde{\tau}_{S,i}\le x\},
    \qquad 0\le x\le 1.
\]
The shape of this distribution summarizes how retained lineage-splitting events are positioned within the full time interval from tumor initiation to sampling. In the results below, we use this distribution to compare branching-time patterns across exponential and density-dependent growth regimes.

\section{Results}

In this section, we present theoretical and numerical results for the distribution of internal branching times under the growth regimes introduced above. We first examine the constant-rate exponential model, which serves as a reference case for unconstrained tumor-cell expansion. Under a large-\(T\), large-\(m\) coalescent point process approximation, we show that the marginal branching-time distribution is unimodal. We then study density-dependent logistic growth. For this case, we derive a deterministic sampled-genealogy approximation that gives an explicit branching-time intensity. This approximation describes how the branching-time distribution changes with the observation time and characterizes a transition in its modal structure: from an expansion-dominated phase, through an early--recent bimodal phase, to a recent-dominated phase.

\subsection{Branching-Time Distributions Under Exponential Growth}

We first analyze the constant-rate exponential-growth model. Recall that \(m\) denotes the number of sampled extant cells. For a fully resolved rooted binary genealogy with \(m\) leaves, the sampled tree contains \(m-1\) internal branching events. We index these events by $\mathcal{I}_S=\{1,\ldots,m-1\}$. Note that the branching events are ordered at random, not in chronological order by when the branching events occurred. As in Section~\ref{Sec: Tree and Branching Time}, let \(\tau_{S,i}\) denote the forward-time location of internal node \(i\in\mathcal{I}_S\), measured on the same time scale as the birth--death process.

Following the large-\(T\), large-\(m\) coalescent point process (CPP)
approximation of \citet{johnson2023clonerate} (see Section~2.2), we
distinguish forward-time branching locations from backward CPP node
depths. Let \(H_{S,i}\) denote the CPP node depth of internal node \(i\),
measured backward from the observation time \(T\). Then
\[
    H_{S,i}=T-\tau_{S,i}.
\]
For a supercritical constant-rate birth--death process with net growth
rate \(r_0=\lambda_0-\mu_0>0\), the CPP approximation gives
\begin{equation}
    H_{S,i}
    =
    T
    -
    \frac{1}{r_0}
    \left[
        \log(1/W_S)+\log m+U_{S,i}
    \right],
    \qquad i\in\mathcal{I}_S,
    \label{eq:exp-cpp-depth}
\end{equation}
where $W_S\sim\operatorname{Exp}(1)$, $U_{S,i}\overset{\mathrm{i.i.d.}}{\sim}\operatorname{Logistic}(0,1)$, and \(W_S\) is independent of the \(U_{S,i}\). Equivalently, the forward-time branching locations satisfy
\begin{equation}
    \tau_{S,i}
    =
    \frac{1}{r_0}
    \left[
        \log(1/W_S)+\log m+U_{S,i}
    \right],
    \qquad i\in\mathcal{I}_S .
    \label{eq:exp-cpp-forward}
\end{equation}
We then have the following unimodality result.

\begin{theorem}\label{Thm: unimodality_exponential}
Under the exponential-growth CPP approximation, the marginal density of each \(\tau_{S,i}\) is log-concave. Consequently, the marginal branching-time distribution is unimodal.
\end{theorem}

Theorem~\ref{Thm: unimodality_exponential} provides a reference pattern for unconstrained exponential growth. Under the CPP approximation, retained internal branching events are distributed around a single characteristic time window. This time scale is determined by
\[
    \frac{1}{r_0}\{\log(1/W_S)+\log m\}.
\]
Thus, the location of the branching-time distribution depends on the net growth rate, the sample size, and the realization-specific factor \(W_S\). The random terms \(U_{S,i}\) describe the spread of internal node times around this characteristic scale. After scaling by the observation time \(T\), the corresponding modal location is of order $\frac{1}{r_0T}\{\log(1/W_S)+\log m\}$. Hence, when \(T\) is large relative to the characteristic sampling scale \((\log m)/r_0\), the mode lies close to the rootward side of the sampled genealogy. 

This rootward single-window structure has a natural genealogical interpretation. During exponential expansion, the population size grows rapidly after establishment. Retained internal nodes in the sampled tree correspond only to division events for which both daughter lineages leave descendants in the sample. Early ancestral divisions have more time to generate descendant lineages that are represented in the final sample. These divisions therefore tend to be retained and form a mode near the root of the sampled genealogy. Once the population is much larger than the sample size, a newly created daughter lineage typically represents a small fraction of the final population. Such recent side branches are less likely to contain sampled descendants on both daughter branches, and are often removed when the full genealogy is pruned to the sampled tree. As a result, retained branching events concentrate in one rootward time window rather than forming multiple modes.

\begin{figure}[htbp]
    \centering

    \includegraphics[width=0.9\textwidth]
    {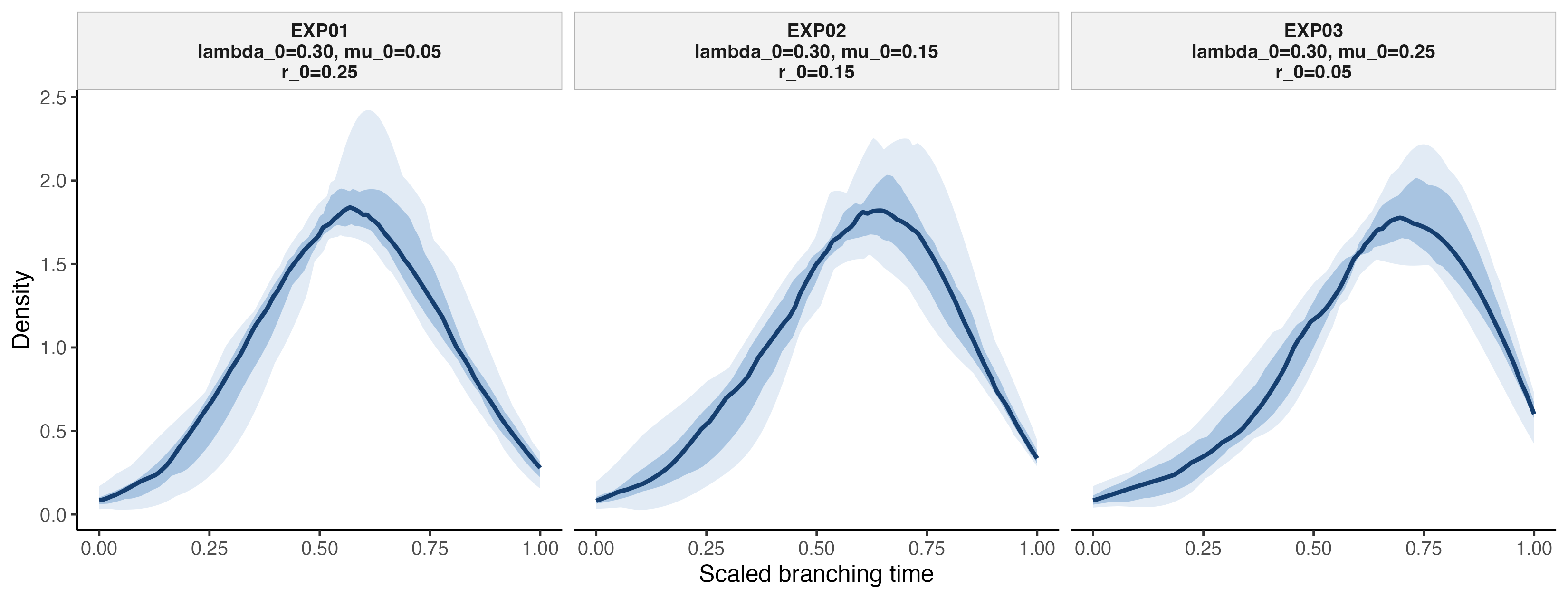}

    \caption{
    Branching-time distributions under constant-rate exponential growth. Densities are computed from internal branching times scaled by the observation time \(T\). The solid curve shows the pointwise median kernel density estimate across $20$ independent simulations. Dark and light bands represent the 25th--75th and 5th--95th percentile ranges, respectively.
    }
    \label{fig:exponential-results}
\end{figure}

We next examine whether the single-mode pattern predicted by the CPP approximation is observed in finite simulated genealogies after reconstruction and \(T\)-scaling of branching times. Simulations were performed with a common birth rate \(\lambda_0=0.30\) and death rates $\mu_0=0.05, 0.15, 0.25$. These parameter choices correspond to net growth rates $r_0=0.25, 0.15, 0.05$, respectively. For each parameter setting, we generated $20$ independent trees.

Across all three parameter settings, the median density of \(T\)-scaled branching times was unimodal (Figure~\ref{fig:exponential-results}). The solid curve shows the pointwise median kernel density estimate across the 20 independent simulations. The darker band shows the pointwise 25th--75th percentile range, and the lighter band shows the pointwise 5th--95th percentile range. These bands summarize simulation-to-simulation variation in the estimated density curves. 

As the death rate increased and the net growth rate decreased, the mode shifted toward larger \(T\)-scaled branching times. The median mode locations were \(0.580\), \(0.646\), and \(0.726\) for net growth rates \(0.25\), \(0.15\), and \(0.05\), respectively. This shift is explained by the time scale in the CPP approximation. For fixed \(T\) and \(m\), decreasing \(r_0\) increases the characteristic forward time $\frac{1}{r_0}\{\log(1/W_S)+\log m\}$. Thus, retained branching events occur later on the \(T\)-scaled time axis when net growth is slower. In these simulations, changing the birth--death parameters affected the relative timing of retained internal nodes, but it did not produce a systematic departure from the single-mode structure predicted under exponential growth.

\subsection{Branching-Time Distributions Under Logistic Growth}

We next derive an analytical approximation for the branching-time distribution under density-dependent growth. The exact finite-population distribution is difficult to obtain because the sampled genealogy depends on the full stochastic population trajectory and on which surviving lineages are represented in the sample. To obtain a tractable approximation, we replace the stochastic population size \(N(t)\) by its deterministic mean-field trajectory. We then analyze the genealogy of sample-represented lineages in this deterministic environment.

Recall that \(r_0=\lambda_0-\mu_0>0\) denotes the baseline net growth rate. For both density-dependent models introduced above, the deterministic population size \(\bar N(t)\) satisfies the logistic equation
\[
    \frac{d\bar N(t)}{dt}
    =
    r_0
    \left(1-\frac{\bar N(t)}{\kappa}\right)
    \bar N(t),
    \qquad \bar N(0)=N_0 .
\]
The solution is $\bar N(t)=\frac{\kappa}{1+c e^{-r_0 t}}$, where $c=\frac{\kappa-N_0}{N_0}$.

We define the time-dependent effective rates induced by this deterministic trajectory as $\bar\lambda(t)=\lambda(\bar N(t))$, $\bar\mu(t)=\mu(\bar N(t))$, $\bar r(t)=\bar\lambda(t)-\bar\mu(t)$. For both logistic models, $\bar r(t)=r_0\left(1-\frac{\bar N(t)}{\kappa}\right)$. Thus, the per-cell net growth rate decreases from approximately \(r_0\) when \(\bar N(t)\ll \kappa\) to \(0\) as \(\bar N(t)\to\kappa\). The two models differ, however, in their absolute birth and death rates along this trajectory. Under density-dependent death, $\bar\lambda(t)=\lambda_0$ and $\bar\mu(t)=\mu_0+\frac{\bar N(t)}{\kappa}r_0$. Under density-dependent proliferation suppression, $\bar\lambda(t)=\lambda_0-\frac{\bar N(t)}{\kappa}r_0$ and $\bar\mu(t)=\mu_0$. Therefore, although the two models have the same deterministic logistic net growth rate, they imply different levels of birth--death turnover. At carrying capacity, $\bar\lambda(\infty)=\bar\mu(\infty)=\lambda_0$ under density-dependent death, whereas $\bar\lambda(\infty)=\bar\mu(\infty)=\mu_0$ under density-dependent proliferation suppression.

Recall that \(T\) denotes the observation time. Suppose that \(m\) extant cells are sampled uniformly from the population at time \(T\). Under the deterministic approximation, the population size at sampling is \(\bar N(T)\). For tractability, we approximate sampling \(m\) cells without replacement by Bernoulli sampling, in which each extant cell is independently sampled with probability $\rho_m(T)=\frac{m}{\bar N(T)}$. Thus, the expected sample size is \(m\). 

For a single lineage alive at time \(t\le T\), define
\[
    q(t;T)
    =
    \Pr\{\text{the lineage has at least one sampled descendant at time }T\}.
\]
Under the deterministic approximation, descendants of this lineage evolve as a time-inhomogeneous birth--death process with per-cell rates \(\bar\lambda(t)\) and \(\bar\mu(t)\). The sample-representation probability satisfies the backward equation
\begin{equation}
    \frac{\partial q}{\partial t}
    =
    -\bar r(t)q(t;T)
    +
    \bar\lambda(t)q(t;T)^2,
    \qquad
    q(T;T)
    =
    \rho_m(T).
    \label{eq:q-backward-density-regulated}
\end{equation}

A retained internal branching event at time \(t\) occurs when a division event produces two daughter lineages that both have sampled descendants at time \(T\). Since the total division rate at time \(t\) is \(\bar\lambda(t)\bar N(t)\), the approximate branching-time intensity is $I_T(t)=\bar\lambda(t)\bar N(t)q(t;T)^2$. The corresponding approximate branching-time density is
\begin{equation}
    f_T(t)
    =
    \frac{I_T(t)}
    {\int_0^T I_T(u)\,du},
    \qquad
    0\le t\le T .
    \label{eq:branching-time-density}
\end{equation}
Because the denominator in \eqref{eq:branching-time-density} is independent of \(t\), the density \(f_T\) and the intensity \(I_T\) have the same modal structure. The following lemma gives the explicit formula for $I_T(t)$.

\begin{lemma}\label{lem:explicit-intensity}
Under the deterministic sampled-genealogy approximation, the branching-time
intensity for either density-dependent logistic model is
\begin{equation}
    I_T(t)
    =
    \frac{
        \displaystyle
        \frac{\lambda(\kappa)}{\kappa}
        +
        \frac{\lambda_0 c}{\kappa}e^{-r_0 t}
    }{
    \left[
        \displaystyle
        \frac{1}{m}
        +
        \frac{\lambda(\kappa)}{\kappa}(T-t)
        +
        \frac{\lambda_0 c}{\kappa r_0}
        \left(e^{-r_0 t}-e^{-r_0 T}\right)
    \right]^2
    },
    \qquad 0\le t\le T,
    \label{eq:explicit-intensity-logistic}
\end{equation}
where
\[
    \lambda(\kappa)
    =
    \begin{cases}
        \lambda_0, & \text{for density-dependent death},\\
        \mu_0, & \text{for density-dependent reduced reproduction}.
    \end{cases}
\]
\end{lemma}

We now use Lemma~\ref{lem:explicit-intensity} to characterize how the modal structure changes with the observation time \(T\). Since \(f_T\) differs from \(I_T\) only by a positive normalizing constant, we study the modes of \(I_T\). Recall that $c=\frac{\kappa-N_0}{N_0}$, and define the dimensionless parameters $A=\frac{\lambda_0 c}{\lambda(\kappa)}$ and $B=\frac{r_0\kappa}{m\lambda(\kappa)}$.

\begin{theorem}\label{thm:sequential-mode-transition}
Assume that $A>20$ and $10<B<\frac{A}{2}$\footnote{Under the single-cell initialization \(N_0=1\), we have $c=\kappa-1$, $A=\frac{\lambda_0(\kappa-1)}{\lambda(\kappa)}$, $B=\frac{r_0\kappa}{m\lambda(\kappa)}$. Thus, for large \(\kappa\) and sparse sampling \(1\ll m\ll\kappa\), the conditions \(A>20\) and $10<B<\frac{A}{2}$ are mild.}. Then there exist transition times $0<T_1<T_2<T_3$ such that：

\begin{enumerate}
    \item If \(0<T<T_1\), then \(I_T\) is strictly increasing on \([0,T]\). Hence the branching-time density has a single boundary mode at the recent endpoint \(t=T\). 

    \item If \(T_1<T<T_2\), then \(I_T\) has exactly one interior local maximum. 

    \item If \(T_2<T<T_3\), then \(I_T\) has one interior local maximum, one interior local minimum, and a boundary maximum at the recent endpoint \(t=T\). Thus the branching-time density is bimodal.

    \item If \(T>T_3\), then \(I_T\) has no interior local maximum and has one interior local minimum. The recent endpoint \(t=T\) is the only asymptotically non-negligible mode. More precisely, for each fixed \(a\ge 0\), $I_T(T-a)\rightarrow\frac{\lambda(\kappa)m^2/\kappa}{\left(1+\lambda(\kappa)ma/\kappa\right)^2}$ as $T\to\infty$, while $\frac{I_T(0)}{I_T(T)}\rightarrow 0$. Moreover, for every \(0<\eta<1\), $\int_0^{(1-\eta)T} f_T(t)\,dt \rightarrow 0$ as $T\to\infty$. Thus, after rescaling by the total observation time, the branching-time distribution concentrates near the recent endpoint.
\end{enumerate}
\end{theorem}

Theorem~\ref{thm:sequential-mode-transition} describes how the sampled-genealogy approximation changes as a density-dependent population is observed for longer time after initiation. The branching-time intensity reflects two sources of retained internal nodes. The first source is the initial expansion phase. During this phase, rapid expansion can establish major ancestral lineages that are represented in the final sample. This mechanism gives rise to an older, expansion-associated interior mode when the observation time is sufficiently long. The second source is the regulated phase near carrying capacity. As \(\bar N(t)\) approaches \(\kappa\), the per-cell net growth rate decreases to zero. However, when \(\lambda(\kappa)>0\), cell division and cell death continue to generate lineage turnover. Division events close to the sampling time can be retained because both daughter lineages need to persist only for a short time before observation. This produces a boundary mode at \(t=T\).

The theorem characterizes the competition between these two sources of branching events. For very short observation times, the sampled tree is shallow and the intensity increases toward the sampling endpoint. For intermediate observation times, an expansion-associated interior mode appears. For longer times, this older mode coexists with a recent boundary mode at \(t=T\), producing an early--recent bimodal profile. Finally, for large \(T\), the recent boundary mode dominates. Thus, prolonged density-dependent regulation leaves a branching-time signature that differs qualitatively from the exponential-growth reference case: the sampled genealogy increasingly reflects recent lineage-splitting events generated during the regulated phase, rather than only the early establishment of ancestral lineages.

\begin{figure}[htbp]
    \centering
    \includegraphics[
        width=\linewidth,
        height=0.68\textheight,
        keepaspectratio
    ]{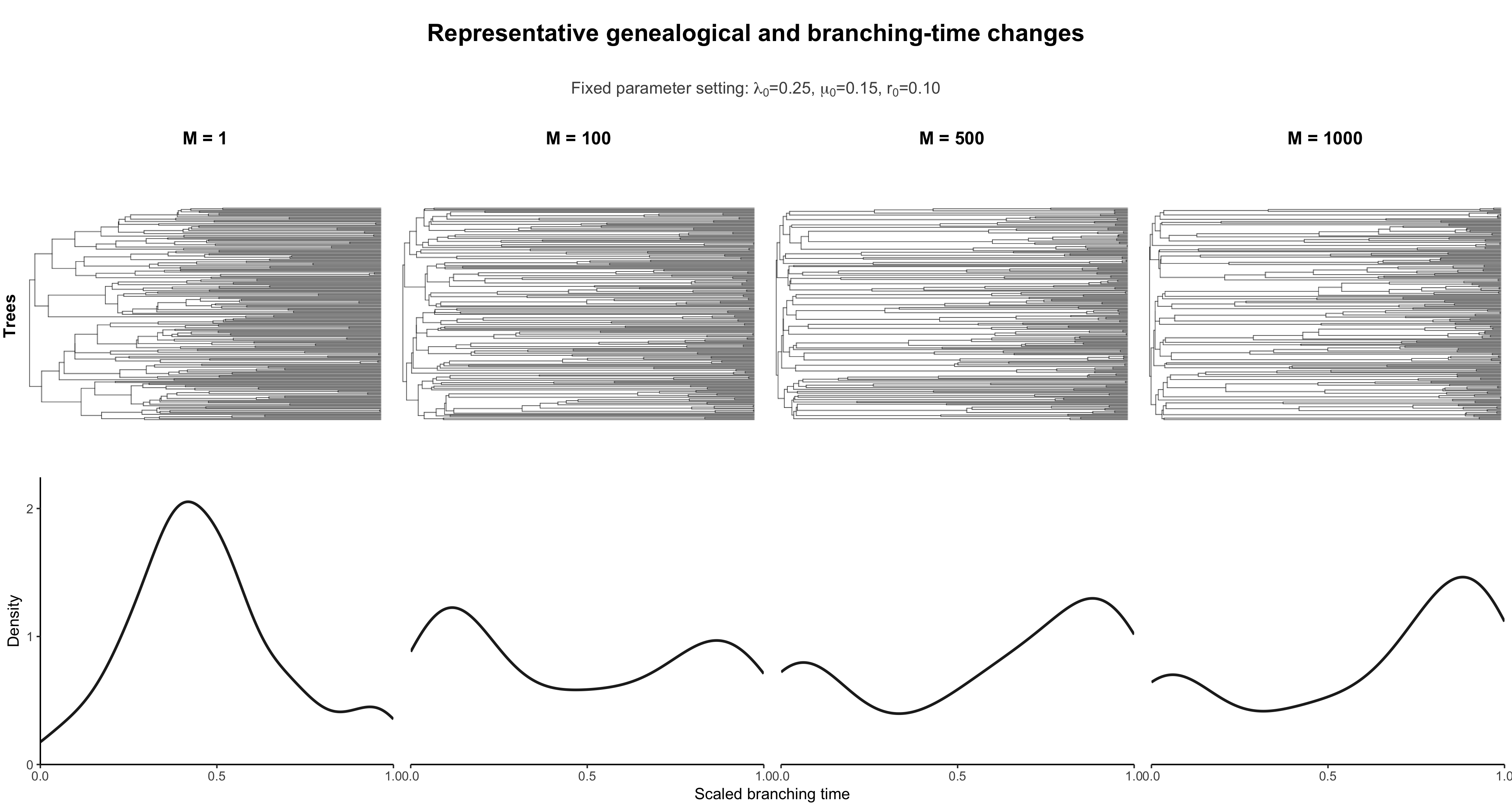}
    \caption{
    Representative temporal changes under density-dependent death with
    \(\lambda_0=0.25\), \(\mu_0=0.15\), and \(r_0=0.10\).
    Columns correspond to runtime multipliers \(M=1,100,500,\) and \(1000\),
    where \(T_M=M T_{\mathrm{base}}\) and \(T_{\mathrm{base}}\) is the first time at which \(N(t)\) reaches \(0.9\kappa\).
    The upper row shows sampled genealogies at each observation time.
    The lower row shows kernel density estimates of the \(T_M\)-scaled internal branching times.
    }
    \label{fig:death-representative}
\end{figure}

We next used finite stochastic simulations to examine the temporal changes described by Theorem~\ref{thm:sequential-mode-transition}. We focus first on the density-dependent death model. For a representative parameter setting, let \(T_{\mathrm{base}}\) denote the first time at which the living population in the simulated trajectory reached \(0.9\kappa\). For a runtime multiplier \(M\), we define the observation time $T_M = M T_{\mathrm{base}}$. Thus, \(M=1\) corresponds to the first observation near carrying capacity, whereas \(M=100\), \(M=500\), and \(M=1000\) correspond to increasingly long periods of density regulation. At each observation time \(T_M\), we sampled extant cells using the same sampling procedure as above and computed the \(T_M\)-scaled internal branching times, $\widetilde{\tau}_{S_M,i}=\frac{\tau_{S_M,i}}{T_M}$.

Figure~\ref{fig:death-representative} shows sampled genealogies and the corresponding branching-time density estimates for \(\lambda_0=0.25\), \(\mu_0=0.15\), and \(r_0=\lambda_0-\mu_0=0.10\). The branching-time profile changed systematically as the runtime multiplier \(M\) increased. At \(M=1\), the density was dominated by an interior mode. By \(M=100\), both an older, rootward mode and a recent, tipward mode were visible. As \(M\) increased to \(500\) and \(1000\), the recent mode became the dominant feature of the distribution. Thus, in this representative simulation, the branching-time distribution changed from an interior-dominated phase to an early--recent bimodal phase, and then to a recent-dominated phase. This is the same qualitative transition characterized by Theorem~\ref{thm:sequential-mode-transition}.

\subsubsection{Comparison between the two density-dependent mechanisms}

The two logistic models have the same deterministic population-size trajectory and the same time-dependent net growth rate,
\[
    \bar r(t)
    =
    r_0\left(1-\frac{\bar N(t)}{\kappa}\right),
\]
but they differ in the mechanism through which density dependence is imposed. Under density-dependent death, the birth rate remains constant while the death rate increases with population size. Under density-dependent proliferation suppression, the death rate remains constant while the birth rate decreases. Consequently, although the two models have identical deterministic growth trajectories, they have different levels of birth--death turnover after the population approaches carrying capacity.

\begin{figure}[htbp]
    \centering

    \begin{subfigure}[t]{0.98\textwidth}
        \centering
        \includegraphics[width=\linewidth]
        {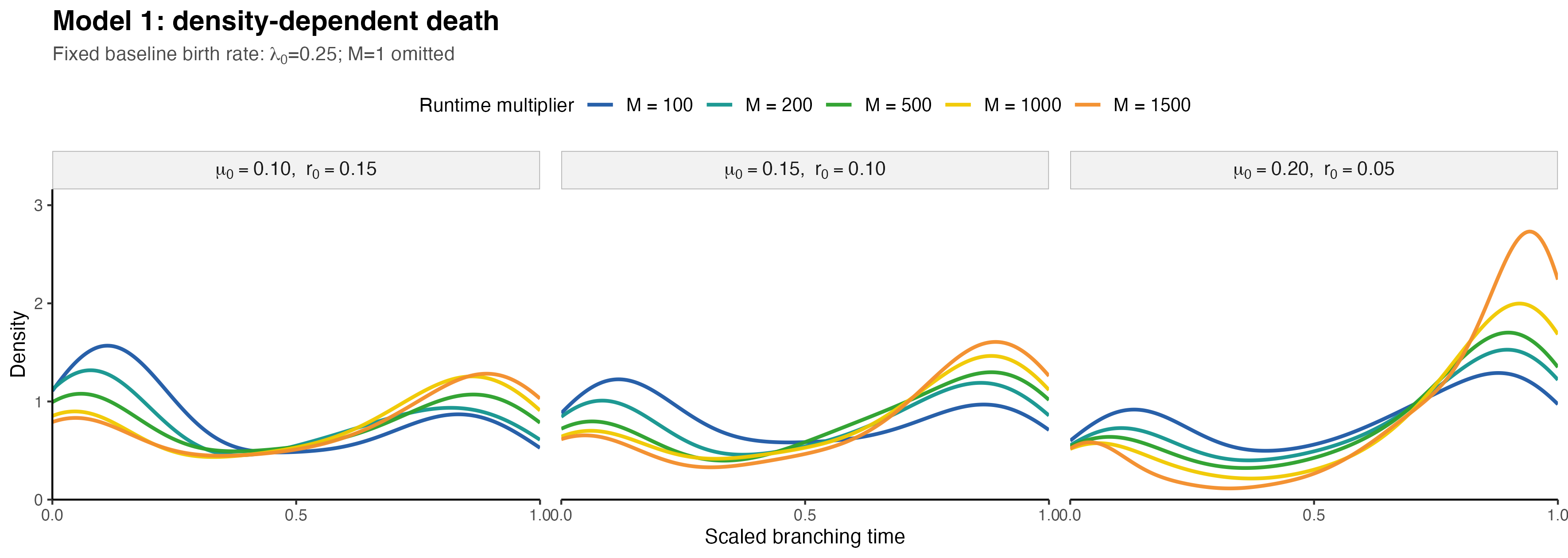}
        \caption{Model 1: density-dependent death.}
        \label{fig:fixed-birth025-model1}
    \end{subfigure}

    \vspace{0.35em}

    \begin{subfigure}[t]{0.98\textwidth}
        \centering
        \includegraphics[width=\linewidth]
        {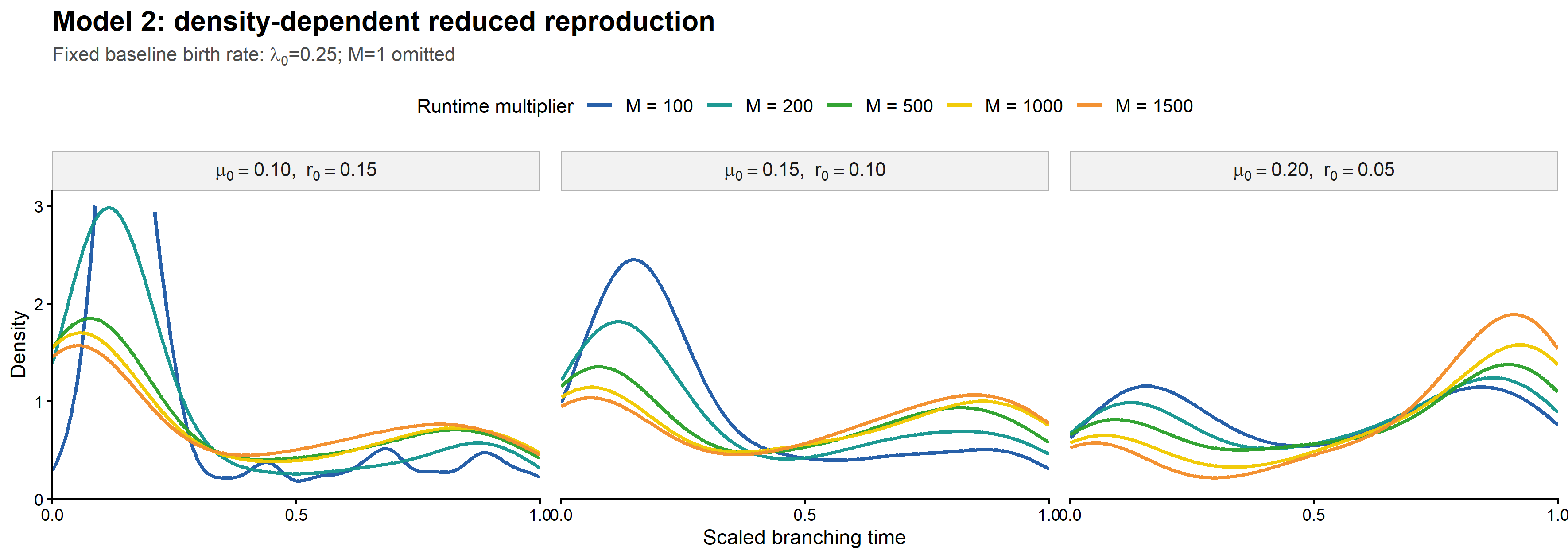}
        \caption{Model 2: density-dependent proliferation suppression.}
        \label{fig:fixed-birth025-model2}
    \end{subfigure}

    \caption{Comparison of branching-time density estimates under the two density-dependent logistic models at fixed baseline birth rate \(\lambda_0=0.25\).
    \textbf{(a)} Model 1, in which density dependence acts through the death rate.
    \textbf{(b)} Model 2, in which density dependence acts through the birth rate.
    In both panels, columns correspond to baseline death rates \(\mu_0=0.10,0.15,\) and \(0.20\), with net growth rates \(r_0=0.15,0.10,\) and \(0.05\), respectively.
    Curves show kernel density estimates of \(T_M\)-scaled internal branching times at observation times \(T_M=M T_{\mathrm{base}}\), for \(M\in\{100,200,500,1000,1500\}\).
    The baseline time \(T_{\mathrm{base}}\) is the first time at which the living population reaches \(0.9\kappa\).}
    \label{fig:fixed-birth025-model-comparison}
\end{figure}

Figure~\ref{fig:fixed-birth025-model-comparison} compares the two mechanisms for a fixed baseline birth rate \(\lambda_0=0.25\). Figures~\ref{fig:fixed-birth025-model-comparison}\subref{fig:fixed-birth025-model1} and~\subref{fig:fixed-birth025-model2} show the density-dependent death and proliferation-suppression models, respectively. In each panel, the three columns correspond to baseline death rates \(\mu_0=0.10\), \(0.15\), and \(0.20\), giving net growth rates \(r_0=0.15\), \(0.10\), and \(0.05\). The curves show kernel density estimates of the \(T_M\)-scaled internal branching times for runtime multipliers $M\in\{100,200,500,1000,1500\}$.

Both models exhibit the qualitative transition predicted by Theorem~\ref{thm:sequential-mode-transition}. As the runtime multiplier increased, branching-time mass shifted toward the recent endpoint, and this shift became more pronounced as the net growth rate decreased. However, the extent of this redistribution differed substantially between the two models. Under density-dependent death, the recent mode strengthened steadily with increasing runtime, while the rootward mode weakened. This effect was particularly pronounced when \(\mu_0=0.20\), where the recent mode became the dominant feature of the distribution. In contrast, under density-dependent proliferation suppression, the redistribution toward recent branching times was noticeably weaker. For \(\mu_0=0.10\) and \(0.15\), a substantial rootward mode persisted even at the longest observation times, and the recent mode became clearly dominant only when \(\mu_0\) approached \(\lambda_0\).

These differences are consistent with the distinct turnover rates of the two models. Near carrying capacity, density-dependent death maintains a high rate of cell division that is balanced by an equally high death rate, producing continuous lineage turnover. In contrast, proliferation suppression reduces both population growth and lineage production by lowering the birth rate. Consequently, although both models share the same deterministic population-size trajectory, the higher turnover under density-dependent death generates substantially more recent retained branching events and produces a stronger shift of the branching-time distribution toward the sampling time.

To summarize branching-time redistribution across the parameter grid, we define two summary statistics. For a sampled tree at observation time \(T_M\), define
\[
E_{0.25}
=
\frac{1}{m-1}
\sum_{i\in\mathcal I_{S_M}}
\mathbf{1}\{\widetilde{\tau}_{S_M,i}\le 0.25\},
\qquad
L_{0.75}
=
\frac{1}{m-1}
\sum_{i\in\mathcal I_{S_M}}
\mathbf{1}\{\widetilde{\tau}_{S_M,i}\ge 0.75\},
\]
and let
\[
D=L_{0.75}-E_{0.25}.
\]
Positive values of \(D\) indicate that branching events are concentrated near the sampled tips, whereas negative values indicate a stronger concentration near the root. The notation \(D\) is used only for this summary statistic; the death rate is denoted by \(\mu_0\).

Figure~\ref{fig:late-minus-early-heatmap-comparison} summarizes \(D\) for both logistic models over a common parameter grid with \(\lambda_0\in\{0.25,0.30,0.35\}\), multiple values of \(\mu_0<\lambda_0\), and runtime multipliers \(M\in\{100,200,500,1000,1500\}\). Under density-dependent death (Figure~\ref{fig:death-contrast}), \(D\) increased with runtime for most parameter settings. Many parameter combinations changed from \(D<0\) to \(D>0\), indicating a progressive redistribution of branching events toward the sampled tips. This transition was most pronounced for larger death rates and smaller net growth rates. Under density-dependent proliferation suppression (Figure~\ref{fig:birth-contrast}), the same qualitative trend was present but substantially weaker. For low and intermediate death rates, \(D\) remained negative even at the longest observation times, indicating persistent rootward concentration of branching events. Positive values of \(D\) occurred mainly when \(\mu_0\) was close to \(\lambda_0\).

Together, Figures~\ref{fig:fixed-birth025-model-comparison} and~\ref{fig:late-minus-early-heatmap-comparison} show that the two logistic models produce distinct branching-time signatures despite having identical deterministic population trajectories. Density-dependent death generates higher birth--death turnover near carrying capacity, producing a stronger redistribution of retained branching events toward the sampling time. In contrast, density-dependent proliferation suppression reduces lineage production during the regulated phase, thereby preserving a larger proportion of branching events generated during the initial expansion.

\begin{figure}[htbp]
    \centering

    \begin{subfigure}[t]{0.98\textwidth}
        \centering
        \includegraphics[
            width=\linewidth,
            height=0.40\textheight,
            keepaspectratio
        ]{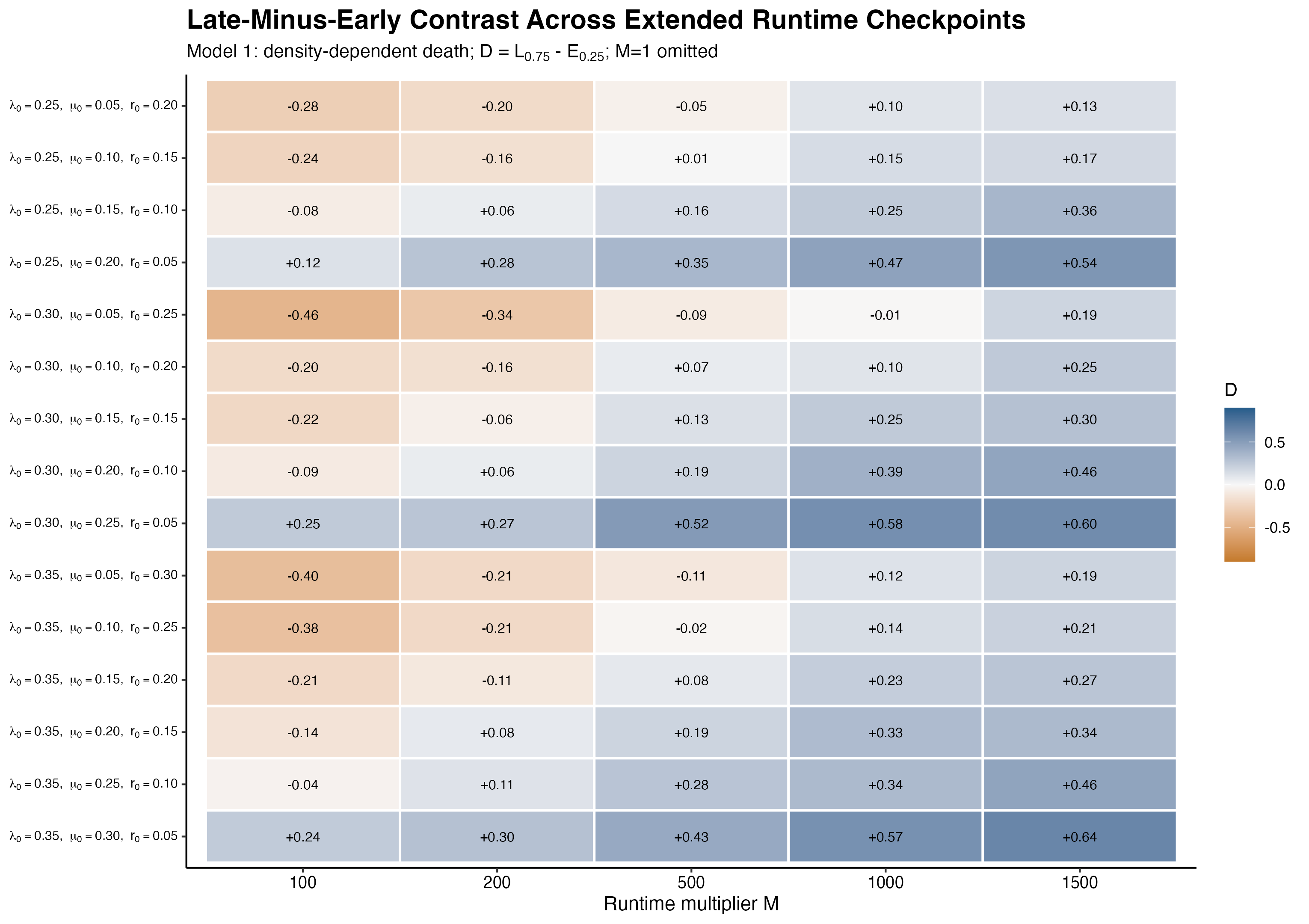}
        \caption{Model 1: density-dependent death.}
        \label{fig:death-contrast}
    \end{subfigure}

    \vspace{0.35em}

    \begin{subfigure}[t]{0.98\textwidth}
        \centering
        \includegraphics[
            width=\linewidth,
            height=0.40\textheight,
            keepaspectratio
        ]{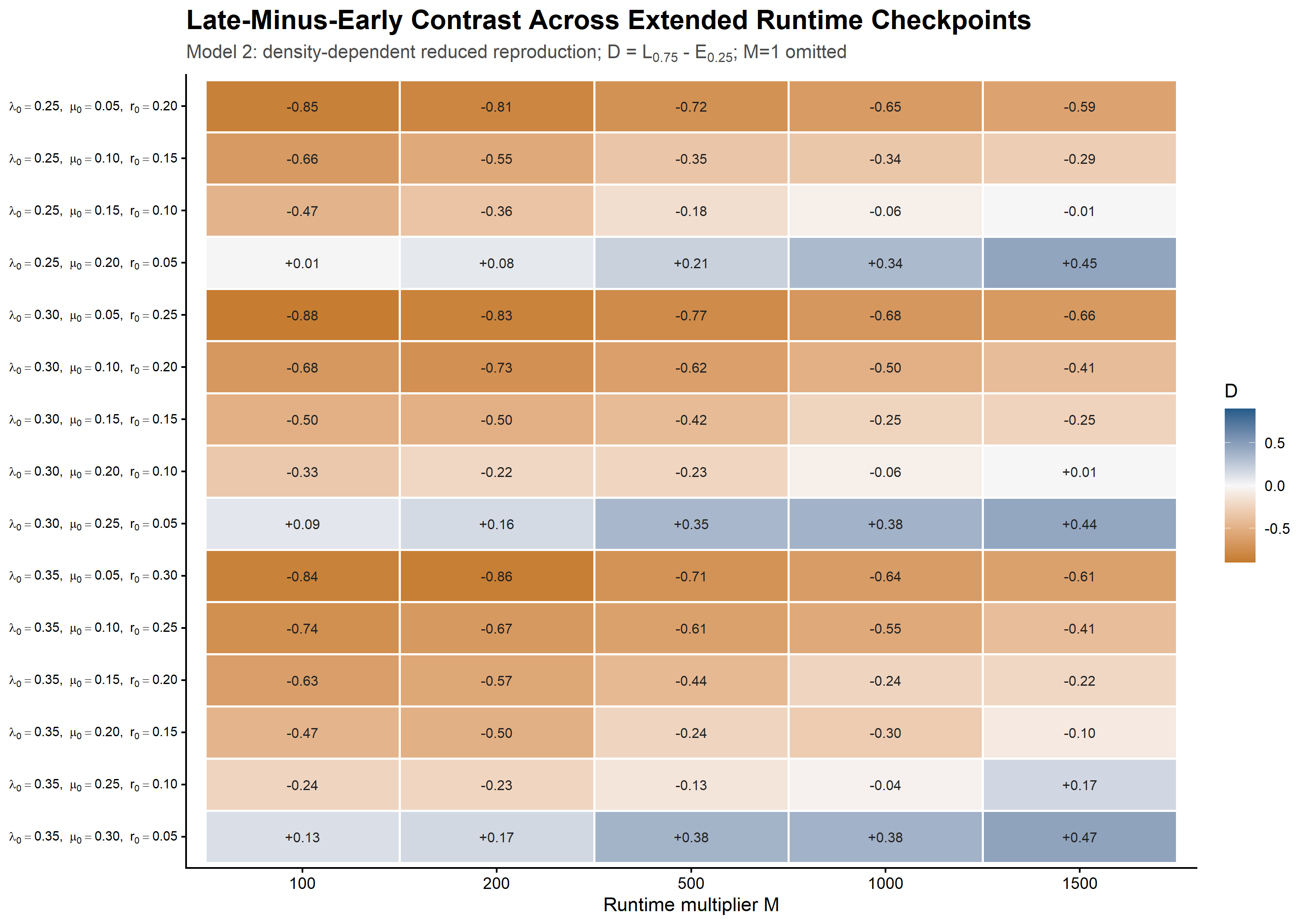}
        \caption{Model 2: density-dependent proliferation suppression.}
        \label{fig:birth-contrast}
    \end{subfigure}

\caption{
Comparison of the late--early branching-time contrast under the two density-dependent logistic models.
\textbf{(a)} Density-dependent death.
\textbf{(b)} Density-dependent proliferation suppression.
Rows correspond to baseline birth--death parameter settings, and columns correspond to runtime multipliers \(M=100,200,500,1000,\) and \(1500\).
Each cell reports $D=L_{0.75}-E_{0.25}$, where \(E_{0.25}\) and \(L_{0.75}\) are the fractions of \(T_M\)-scaled branching times in the intervals \([0,0.25]\) and \([0.75,1]\), respectively.
Positive values (\(D>0\), blue) indicate greater concentration of branching events near the sampled tips, whereas negative values (\(D<0\), orange) indicate greater concentration near the root.
}
    \label{fig:late-minus-early-heatmap-comparison}
\end{figure}

\section{Concluding Comments}\label{Sec:conclusion}

In this work, we examined how the distribution of internal branching times changes across different birth--death growth regimes. Specifically, we compared constant-rate exponential growth with two density-dependent logistic models, in which regulation acts through either increased death or reduced proliferation.

Under exponential growth, the large-\(T\), large-\(m\) coalescent point process approximation yields a unimodal marginal branching-time distribution. Its mode lies on the rootward side of the sampled genealogy when the observation time is large relative to the characteristic sampling scale. Density-dependent logistic growth produced a different sequence of branching-time patterns. Under the deterministic sampled-genealogy approximation, the branching-time intensity changes from an interior expansion-associated mode, to an early--recent bimodal phase, and finally to a recent-dominated phase. 

The comparison between the two logistic models shows that the net-growth rate alone does not determine the sampled genealogy. Although the models have the same logistic population trajectory, they differ in their birth--death turnover near carrying capacity. Density-dependent death maintains a higher division rate during the regulated phase and produces a larger redistribution of branching times toward the recent endpoint. Density-dependent proliferation suppression generates fewer late division events and retains a greater contribution from branching events established during the initial expansion. These results show that internal branching-time distributions can contain qualitative signatures associated with both population growth and the mechanism of density regulation. They should not, however, be interpreted as providing a unique classification of growth regime. Different biological processes may produce similar genealogical patterns, and the present models are simplified reference cases. Although the present work is primarily descriptive rather than inferential, the results suggest that branching-time distributions may serve as useful summary statistics for future statistical inference. Rather than reconstructing an entire population-size trajectory, one could compare observed branching-time distributions with those predicted under competing evolutionary models, or estimate biologically meaningful quantities such as the duration of density regulation or the level of birth--death turnover. Developing statistically rigorous procedures for model selection and parameter estimation based on branching-time distributions remains an important direction for future work.

Several limitations remain. The theoretical results rely on asymptotic and deterministic approximations, and the simulations assume uniform sampling from the extant population. Real tumors may exhibit spatial structure, clonal selection, heterogeneous birth and death rates, treatment effects, and biased sampling. Phylogenetic reconstruction and time calibration also introduce uncertainty. Future work could develop stochastic results for sampled genealogies under density-dependent birth--death processes and examine the robustness of the branching-time patterns to spatial structure, selection, treatment, and nonuniform sampling. Statistical inference procedures based on branching-time distributions could also be developed to compare competing growth models or estimate parameters governing density regulation and birth--death turnover while accounting for sampling and phylogenetic uncertainty.

% =========================
% Appendix: Unimodality under exponential growth (CPP model)
% =========================

\newpage

\appendix

\section{Proof of Theorem \ref{Thm: unimodality_exponential}}

\begin{proof}
Under the large-\(T\), large-\(m\) CPP approximation in \eqref{eq:exp-cpp-forward}, the forward-time branching location is $\tau_{s,i}=\frac{1}{r_0}\left[\log(1/W_s)+\log m+U_{s,i}\right]$, where \(W_s\sim\operatorname{Exp}(1)\), \(U_{s,i}\sim\operatorname{Logistic}(0,1)\), and \(W_s\) is independent of \(U_{s,i}\). We first prove log-concavity for the marginal density of  \(\tau_{s,i}\).

Define $Y_s=\frac{1}{r_0}\left[\log(1/W_s)+\log m\right]$ and $V_{s,i}=\frac{1}{r_0}U_{s,i}$. Then $\tau_{s,i}=Y_s+V_{s,i}$, and \(Y_s\) is independent of \(V_{s,i}\). We first show that \(Y_s\) has a log-concave density. Since $Y_s=\frac{1}{r_0}(\log m-\log W_s)$, we have $W_s = m e^{-r_0Y_s}$. Using the density \(f_W(w)=e^{-w}\) for \(w>0\), and the change of variables \(w=m e^{-r_0 y}\), we obtain $f_Y(y)=r_0 m e^{-r_0 y}\exp\left\{-m e^{-r_0 y}\right\}$ for $y\in\mathbb{R}$. Thus
\[
    \log f_Y(y)
    =
    \log r_0+\log m-r_0 y-m e^{-r_0 y}.
\]
Differentiating twice gives $\frac{d^2}{dy^2}\log f_Y(y)=-m r_0^2 e^{-r_0 y}<0$. Therefore \(f_Y\) is strictly log-concave.

Next, we show that the scaled logistic term \(V_{s,i}\) also has a log-concave density. Since \(U_{s,i}\sim\operatorname{Logistic}(0,1)\) and \(V_{s,i}=U_{s,i}/r_0\), the density of \(V_{s,i}\) is $f_V(v)=r_0\,\frac{e^{r_0v}}{(1+e^{r_0v})^2}$ for $v\in\mathbb{R}$. Its log-density is
\[
    \log f_V(v)
    =
    \log r_0 + r_0v - 2\log(1+e^{r_0v}).
\]
A second differentiation gives $\frac{d^2}{dv^2}\log f_V(v)=-\frac{2r_0^2e^{r_0v}}{(1+e^{r_0v})^2}<0$. Hence \(f_V\) is strictly log-concave.

Since \(Y_s\) and \(V_{s,i}\) are independent, the marginal density of \(\tau_{s,i}=Y_s+V_{s,i}\) is the convolution
\[
    f_\tau(x)
    =
    \int_{-\infty}^{\infty}
    f_Y(x-v)f_V(v)\,dv .
\]
Both \(f_Y\) and \(f_V\) are log-concave densities on \(\mathbb{R}\), and log-concavity is preserved under convolution \citep{SaumardWellner2014}. Therefore \(f_\tau\) is log-concave. Every log-concave density on \(\mathbb{R}\) is unimodal, so the marginal density of \(\tau_{s,i}\) is unimodal.
\end{proof}

\section{Proof of Lemma \ref{lem:explicit-intensity}}

\begin{proof}
Fix the observation time \(T\), and write \(q(t)=q(t;T)\). Let $v(t)=\frac{1}{q(t)}$. From the backward equation
\[
    \frac{dq}{dt}
    =
    -\bar r(t)q(t)+\bar\lambda(t)q(t)^2,
\]
we obtain $\frac{dv}{dt}=\bar r(t)v(t)-\bar\lambda(t)$. The deterministic population trajectory satisfies $\frac{d\bar N(t)}{dt}=\bar r(t)\bar N(t)$. Therefore, $\frac{d}{dt}\left(\frac{v(t)}{\bar N(t)}\right)=-\frac{\bar\lambda(t)}{\bar N(t)}$. Using the terminal condition $q(T;T)=\frac{m}{\bar N(T)}$, we have $\frac{v(T)}{\bar N(T)}=\frac{1}{m}$. Integrating from \(t\) to \(T\) gives $\frac{v(t)}{\bar N(t)}=\frac{1}{m}+\int_t^T\frac{\bar\lambda(u)}{\bar N(u)}\,du$. Equivalently,
\[
    q(t;T)
    =
    \frac{1}
    {
        \bar N(t)
        \left[
            \frac{1}{m}
            +
            \int_t^T
            \frac{\bar\lambda(u)}{\bar N(u)}\,du
        \right]
    }.
\]
Substituting this expression into $I_T(t)=\bar\lambda(t)\bar N(t)q(t;T)^2$ yields
\[
    I_T(t)
    =
    \frac{\bar\lambda(t)/\bar N(t)}
    {
        \left[
            \frac{1}{m}
            +
            \int_t^T
            \frac{\bar\lambda(u)}{\bar N(u)}\,du
        \right]^2
    }.
\]

It remains to evaluate this expression along the deterministic logistic
trajectory. For both density-regulated logistic models, $\bar N(t)=\frac{\kappa}{1+c e^{-r_0t}}$, where $c=\frac{\kappa-N_0}{N_0}$. Moreover, the effective birth rate can be written in the common form $\bar\lambda(t)=\lambda_0-\frac{\bar N(t)}{\kappa}\bigl(\lambda_0-\lambda(\kappa)\bigr)$, where \(\lambda(\kappa)=\lambda_0\) for density-dependent death and \(\lambda(\kappa)=\mu_0\) for density-dependent reduced reproduction. Hence $\frac{\bar\lambda(t)}{\bar N(t)}=\frac{\lambda_0}{\bar N(t)}-\frac{\lambda_0-\lambda(\kappa)}{\kappa}$. Since $\frac{1}{\bar N(t)}=\frac{1+c e^{-r_0t}}{\kappa}$, we obtain
\[
    \frac{\bar\lambda(t)}{\bar N(t)}
    =
    \frac{\lambda(\kappa)}{\kappa}
    +
    \frac{\lambda_0 c}{\kappa}e^{-r_0t}.
\]
Therefore,
\[
    \int_t^T
    \frac{\bar\lambda(u)}{\bar N(u)}\,du
    =
    \frac{\lambda(\kappa)}{\kappa}(T-t)
    +
    \frac{\lambda_0 c}{\kappa r_0}
    \left(e^{-r_0t}-e^{-r_0T}\right).
\]
Substituting this identity into the preceding expression for \(I_T(t)\) gives
\[
    I_T(t)
    =
    \frac{
        \displaystyle
        \frac{\lambda(\kappa)}{\kappa}
        +
        \frac{\lambda_0 c}{\kappa}e^{-r_0 t}
    }{
    \left[
        \displaystyle
        \frac{1}{m}
        +
        \frac{\lambda(\kappa)}{\kappa}(T-t)
        +
        \frac{\lambda_0 c}{\kappa r_0}
        \left(e^{-r_0 t}-e^{-r_0 T}\right)
    \right]^2
    },
    \qquad 0\le t\le T.
\]
This proves the lemma.
\end{proof}

\section{Proof of Theorem \ref{thm:sequential-mode-transition}}

\begin{proof}
Recall that $A=\frac{\lambda_0 c}{\lambda(\kappa)}$ and $B=\frac{r_0\kappa}{m\lambda(\kappa)}$. Since \(f_T\) differs from \(I_T\) only by a positive normalizing constant, \(I_T\) and \(f_T\) have the same modal structure. It is therefore enough to analyze the modes of \(I_T\).

From Lemma~\ref{lem:explicit-intensity},
\[
    I_T(t)
    =
    \frac{
        \displaystyle
        \frac{\lambda(\kappa)}{\kappa}
        +
        \frac{\lambda_0 c}{\kappa}e^{-r_0 t}
    }{
    \left[
        \displaystyle
        \frac{1}{m}
        +
        \frac{\lambda(\kappa)}{\kappa}(T-t)
        +
        \frac{\lambda_0 c}{\kappa r_0}
        \left(e^{-r_0 t}-e^{-r_0 T}\right)
    \right]^2
    },
    \qquad 0\le t\le T .
\]
Multiplying the denominator inside the square brackets by \(\kappa\) does not affect derivatives of the logarithm. Thus,
\[
    \frac{d}{dt}\log I_T(t)
    =
    -\frac{r_0\lambda_0 c e^{-r_0t}}
    {\lambda(\kappa)+\lambda_0 c e^{-r_0t}}
    +
    \frac{
        2\left(\lambda(\kappa)+\lambda_0 c e^{-r_0t}\right)
    }{
        \frac{\kappa}{m}
        +
        \lambda(\kappa)(T-t)
        +
        \frac{\lambda_0 c}{r_0}
        \left(e^{-r_0t}-e^{-r_0T}\right)
    } .
\]
An interior stationary point therefore satisfies
\[
    \frac{\kappa}{m}
    +
    \lambda(\kappa)(T-t)
    +
    \frac{\lambda_0 c}{r_0}
    \left(e^{-r_0t}-e^{-r_0T}\right)
    =
    \frac{
        2\left(\lambda(\kappa)+\lambda_0 c e^{-r_0t}\right)^2
    }{
        r_0\lambda_0 c e^{-r_0t}
    } .
\]
Define
\[
\begin{aligned}
    \Delta_T(t)=
    \frac{\kappa}{m}
    +
    \lambda(\kappa)(T-t)
    +
    \frac{\lambda_0 c}{r_0}
    \left(e^{-r_0t}-e^{-r_0T}\right)
    -
    \frac{
        2\left(\lambda(\kappa)+\lambda_0 c e^{-r_0t}\right)^2
    }{
        r_0\lambda_0 c e^{-r_0t}
    } .
\end{aligned}
\]
Then $\operatorname{sign}\{I_T'(t)\}=-\operatorname{sign}\{\Delta_T(t)\}$. A direct differentiation gives $\Delta_T''(t)=-r_0\lambda_0 c e^{-r_0t}-\frac{2r_0\lambda(\kappa)^2}{\lambda_0 c}e^{r_0t}<0$. Hence \(\Delta_T\) is strictly concave on \([0,T]\). In particular, \(\Delta_T\) has at most two zeros, so \(I_T\) has at most two interior stationary points.

We next determine how the endpoint signs change with \(T\). At the recent endpoint \(t=T\),
\[
    \frac{r_0}{\lambda(\kappa)}\Delta_T(T)
    =
    B
    -
    2\left(
        Ae^{-r_0T}
        +
        2
        +
        A^{-1}e^{r_0T}
    \right).
\]
The right-hand side is negative at \(T=0\), tends to \(-\infty\) as \(T\to\infty\), and has a unique maximum at $T=\frac{\log A}{r_0}$. The maximum value is \(B-8>0\), since \(B>10\). Therefore there exist exactly two positive times \(T_1<T_2\) such that $B=2\left(Ae^{-r_0T}+2+A^{-1}e^{r_0T}\right)$. Consequently,
\[
    I_T'(T^-)
    \begin{cases}
        >0, & 0<T<T_1,\\
        <0, & T_1<T<T_2,\\
        >0, & T>T_2 .
    \end{cases}
\]

At the early endpoint \(t=0\),
\[
    \frac{r_0}{\lambda(\kappa)}\Delta_T(0)
    =
    B+r_0T+A\left(1-e^{-r_0T}\right)-2A-4-\frac{2}{A}.
\]
This expression is strictly increasing in \(T\), because its derivative is $r_0+r_0Ae^{-r_0T}>0$. At \(T=0\), it is negative under the assumptions \(B<A/2\) and \(A>20\), while it tends to \(+\infty\) as \(T\to\infty\). Hence there exists a unique positive time \(T_3\) such that
\[
    B+r_0T_3+A\left(1-e^{-r_0T_3}\right)
    =
    2A+4+\frac{2}{A}.
\]
Thus,
\[
    I_T'(0^+)
    \begin{cases}
        >0, & T<T_3,\\
        <0, & T>T_3 .
    \end{cases}
\]

It remains to show that \(T_2<T_3\). Since \(T_2\) is the later solution of the recent-endpoint equation, we have $Ae^{-r_0T_2}<1$ and $B=2\left(Ae^{-r_0T_2}+2+A^{-1}e^{r_0T_2}\right)$. The condition \(B<A/2\) implies $Ae^{-r_0T_2}>\frac{4}{A}$. Indeed, the function \(x\mapsto 2(x+2+x^{-1})\) is strictly decreasing on \((0,1)\), and
\[
    2\left(\frac{4}{A}+2+\frac{A}{4}\right)
    =
    \frac{A}{2}+4+\frac{8}{A}
    >
    \frac{A}{2}
    >
    B .
\]
Therefore, $A>2e^{r_0T_2/2}$. Evaluating the early-endpoint expression at \(T=T_2\), and using the equation defining \(T_2\), gives
\[
    \frac{r_0}{\lambda(\kappa)}\Delta_{T_2}(0)
    =
    Ae^{-r_0T_2}
    +
    2A^{-1}e^{r_0T_2}
    +
    r_0T_2
    -
    A
    -
    \frac{2}{A}.
\]
Thus \(\Delta_{T_2}(0)<0\) is equivalent to $A\left(1-e^{-r_0T_2}\right)-\frac{2}{A}\left(e^{r_0T_2}-1\right)>r_0T_2$. For fixed \(T_2\), the left-hand side is increasing as a function of \(A\). Since \(A>2e^{r_0T_2/2}\), it is larger than its value at \(A=2e^{r_0T_2/2}\), namely $2\sinh\left(\frac{r_0T_2}{2}\right)$. Because $2\sinh\left(\frac{x}{2}\right)>x$ for every $x>0$, we obtain \(\Delta_{T_2}(0)<0\). Since \(\Delta_T(0)\) is strictly increasing
in \(T\), the unique zero \(T_3\) of \(\Delta_T(0)\) must satisfy $T_2<T_3$.

We now classify the modal structure in the four regimes. First, suppose \(0<T<T_1\). The first transition time satisfies $Ae^{-r_0T_1}>2$. Indeed, at \(T=T_1\) the quantity \(Ae^{-r_0T_1}\) is greater than one, and \(B>10>2(2+2+1/2)\). Hence $T_1<\frac{\log(A/2)}{r_0}$. For \(0\le t\le T<T_1\),
\[
    \Delta_T'(t)
    =
    \lambda(\kappa)
    \left(
        Ae^{-r_0t}
        -
        1
        -
        2A^{-1}e^{r_0t}
    \right)
    >0 .
\]
Thus \(\Delta_T\) is increasing on \([0,T]\). Since \(\Delta_T(T)<0\) for \(T<T_1\), it follows that \(\Delta_T(t)<0\) for all \(0\le t\le T\). Hence \(I_T'(t)>0\) on \((0,T)\), so \(I_T\) is strictly increasing and has a single boundary mode at \(t=T\).

Next, suppose \(T_1<T<T_2\). Since \(T<T_2<T_3\), we have $\Delta_T(0)<0$ and $\Delta_T(T)>0$. Because \(\Delta_T\) is strictly concave, these endpoint signs imply that \(\Delta_T\) has exactly one zero in \((0,T)\). Therefore \(I_T'\) changes from positive to negative at this zero, and \(I_T\) has exactly one interior local maximum.

Now suppose \(T_2<T<T_3\). Then $\Delta_T(0)<0$ and $\Delta_T(T)<0$. However, evaluating at the fixed interior point \(t=T_2\) gives
\[
\begin{aligned}
    \Delta_T(T_2)
    =
    \Delta_{T_2}(T_2)
    +
    \lambda(\kappa)(T-T_2)
    +
    \frac{\lambda_0c}{r_0}
    \left(e^{-r_0T_2}-e^{-r_0T}\right) >0,
\end{aligned}
\]
because \(\Delta_{T_2}(T_2)=0\). Hence \(\Delta_T\) is negative at both endpoints and positive at an interior point. Strict concavity then implies that \(\Delta_T\) has exactly two zeros. Therefore, the sign of \(I_T'\) changes from $+$ to $-$ and then to $+$. Thus \(I_T\) has one interior local maximum, followed by one interior local minimum, and the recent endpoint \(t=T\) is a boundary maximum. The branching-time density is therefore bimodal.

Finally, suppose \(T>T_3\). Then $\Delta_T(0)>0$ and $\Delta_T(T)<0$. Strict concavity implies that \(\Delta_T\) has exactly one zero in \((0,T)\). Consequently, the sign of \(I_T'\) changes from $-$ to $+$. Thus \(I_T\) has one interior local minimum and no interior local maximum. The recent endpoint \(t=T\) is a boundary maximum. The early endpoint \(t=0\) is also a boundary maximum, but it becomes asymptotically negligible. Indeed, direct substitution gives $I_T(T)=\frac{\lambda(\kappa)/\kappa+\left(\lambda_0 c/\kappa\right) e^{-r_0T}}{(1/m)^2}\rightarrow\frac{\lambda(\kappa)m^2}{\kappa}$, whereas $I_T(0)=O(T^{-2})$. Therefore, $\frac{I_T(0)}{I_T(T)}\rightarrow 0$. For each fixed \(a\ge 0\),
\[
    I_T(T-a)
    \rightarrow
    \frac{\lambda(\kappa)/\kappa}
    {
        \left(
            \frac{1}{m}
            +
            \frac{\lambda(\kappa)}{\kappa}a
        \right)^2
    }
    =
    \frac{\lambda(\kappa)m^2/\kappa}
    {
        \left(1+\lambda(\kappa)ma/\kappa\right)^2
    } .
\]

It remains to prove the concentration statement. From the derivation of Lemma~\ref{lem:explicit-intensity}, $I_T(t)=\frac{\bar\lambda(t)/\bar N(t)}{\left[\frac{1}{m}+\int_t^T\frac{\bar\lambda(u)}{\bar N(u)}\,du\right]^2}$. Therefore, $I_T(t)=\frac{d}{dt}\left[\frac{1}{\frac{1}{m}+\int_t^T\frac{\bar\lambda(u)}{\bar N(u)}\,du}\right]$. Hence, $\int_a^b I_T(t)\,dt=\left[\frac{1}{\frac{1}{m}+\int_t^T\frac{\bar\lambda(u)}{\bar N(u)}\,du}\right]_{t=a}^{t=b}$. Taking \(a=0\) and \(b=T\), we obtain $\int_0^T I_T(t)\,dt\rightarrow m$, because the denominator at \(t=T\) is \(1/m\), while the denominator at \(t=0\) diverges as \(T\to\infty\). For any fixed \(0<\eta<1\), $\int_0^{(1-\eta)T} I_T(t)\,dt \rightarrow 0$, since the denominator at \(t=(1-\eta)T\) contains the term $\frac{\lambda(\kappa)}{\kappa}\eta T$, which diverges as \(T\to\infty\). Dividing by
\(\int_0^T I_T(t)\,dt\) yields
\[
    \int_0^{(1-\eta)T} f_T(t)\,dt
    \rightarrow 0 .
\]
Thus, after rescaling by the total observation time, the branching-time
distribution concentrates near the recent endpoint. This completes the proof.
\end{proof}

\begin{singlespace}

\bibliography{bibliography}

\end{singlespace}

\end{document}